\documentclass[review,a4paper]{article}
\usepackage[T1]{fontenc}
\usepackage[utf8]{inputenc}
\usepackage{subfigure}
\usepackage{stmaryrd}
\usepackage[vlined, ruled, linesnumbered]{algorithm2e}
\usepackage{amssymb,amsmath,amsfonts,amsthm}
\usepackage{tikz}

\def\algname{\texttt{Solve-GLTC} }
\def\0{{\bar 0}}

\def\a{{\bf a}}
\def\b{{\bf b}}

\def\p{{\bf p}}

\def\w{{\mathbf w}}
\def\vv{{\mathbf v}}

\newtheorem{theorem}{Theorem}

\newtheorem{lemma}{Lemma}
\newtheorem{corollary}{Corollary}

\begin{document}

\title{An Exact Algorithm for the Generalized List $T$-Coloring Problem}

\author{
Konstanty Junosza-Szaniawski, Pawe{\l} Rz{\k a}{\. z}ewski  \\
\texttt{\{k.szaniawski,p.rzazewski\}@mini.pw.edu.pl} \\ \\
Warsaw University of Technology \\
Faculty of Mathematics and Information Science,\\
Koszykowa 75 , 00-662 Warsaw, Poland
}
\date{ }
\maketitle

\begin{abstract}
The generalized list $T$-coloring is a common generalization of many graph coloring models, including classical coloring, $L(p,q)$-labeling, channel assignment and $T$-coloring. Every vertex from the input graph has a list of permitted labels. Moreover, every edge has a set of forbidden differences. We ask for such a labeling of vertices of the input graph with natural numbers, in which every vertex gets a label from its list of permitted labels and the difference of labels of the endpoints of each edge does not belong to the set of forbidden differences of this edge. In this paper we present an exact algorithm solving this problem, running in time $\mathcal{O}^*((\tau+2)^n)$, where $\tau$ is the maximum forbidden difference over all edges of the input graph and $n$ is the number of its vertices.
Moreover, we show how to improve this bound if the input graph has some special structure, e.g. a bounded maximum degree, no big induced stars or a perfect matching.
\end{abstract}

\section{Introduction}
Probably no graph-theoretic problem received as much attention from discrete mathematics and theoretical computer science community, as graph coloring. Various graph coloring models are extensively studied both for their practical motivations (e.g. in frequency assignment or scheduling), and for their interesting theoretical and structural properties. In the multitude of coloring (or labeling) models, some are similar to each other, while others require completely different techniques and approaches.

A natural tendency is to look for the similarities between different models and try to unify and generalize them. In this spirit, Fiala, Kr\'al' and \v{S}rekovski \cite{FKS} defined and studied the so-called generalized list $T$-coloring problem. 
An instance of a the generalized list $T$-coloring problem is a triple $(G,\Lambda,t)$, where $G=(V,E)$ is a graph and $\Lambda$ and $t$ are functions. Each vertex $v$ of $G$ has a {\em list of permitted labels} $\Lambda(v) \subseteq \mathbb{N}$. Moreover, each edge $e$ has a list of {\em forbidden differences} $t(e) \subseteq \mathbb{N} \cup \{0\}$ over that edge. Moreover, we assume that $0 \in t(e)$ for any $e \in E$. We aim to find a labeling $\varphi \colon V \to \mathbb{N}$, such that $\varphi(v) \in \Lambda(v)$ for every $v \in V$ and $|\varphi(v) - \varphi(w)| \notin t(vw)$ for any edge $vw \in E$.

The notion of the generalized list $T$-coloring unifies many well-studied graph problems. Here we present some of them and describe how to see them as a special case of our problem.

\noindent \textbf{Graph coloring}

Although the graph coloring problem dates back to XIX century, it still raises a considerable attention from many researchers. See the book by Jensen and Toft \cite{JT} about some information about the history and many still open problems in graph coloring.
In the graph coloring problem, we want to assign colors (usually represented by natural numbers) to vertices of the graphs in such a way, that every pair of neighbors receive different colors. The list version of this problem has been introduced independently by Vizing \cite{Vizing} and by Erd\"os, Rubin and Taylor \cite{ERT}. Here each vertex has a list of colors and its color has to be chosen from this list (a non-list coloring is a special case of the list coloring with all lists equal). Therefore an instance of the list coloring is a pair $(G,\Lambda)$, where $G$ is a graph and $\Lambda$ is the function assigning the lists to vertices.

To see that list coloring is a special case of the generalized list $T$-coloring, consider an instance $(G,\Lambda)$ of the list -coloring and let $(G,\Lambda,t)$ be an instance of our problem, where $t(e) = \{0\}$ for all $e \in E$.

\noindent \textbf{$L(p,q)$-labeling}

The notion $L(p,q)$-labeling (for integers $p \geq q \geq 1$) is inspired by the frequency assignment problem in telecommunication. We look for a labeling of the vertices $G$ with integers, so that the labels of adjacent vertices differ by at least $p$ and the labels of the vertices with a common neighbor differ by at least $q$. The best studied case is $L(2,1)$-labeling, inroduced by Griggs and Yeh \cite{GY}. We refer the reader to surveys on this problem by Yeh \cite{Yeh} and Calamoneri \cite{Calamoneri}. The list version of this problem has been studied for example by Fiala and \v{S}krekovski \cite{FS}.

Notice that (list) $L(p,q)$-labeling of a graph $G$ is equivalent to an instance $(G^2,\Lambda,t)$ of the generalized list $T$-coloring, where $G^2$ is a square of $G$, $\Lambda(v)$ is the list of labels available for $v$ (or the set of all colors in a non-list case) and $t(e)=\{0,1,..,p-1\}$ for $e \in E(G)$ and $t(e)=\{0,1,..,q-1\}$ for $e \in E(G^2) \setminus E(G)$.

\noindent \textbf{Channel assignment}

Channel assignment problem is a generalization of $L(p,q)$-labeling. Its instance is a pair $(G,\omega)$, where $G$ is a graph and $\omega$ is a weight function $\omega \colon E(G) \to \mathbb{N}$. We ask for such a labeling $f$ of vertices of $G$ with natural numbers, that $|f(v) - f(u)| \geq \omega(uv)$ for any $uv \in E(G)$. We refer the reader to the survey by Kr\'{a}l' \cite{Kral-surv} for more information about this problem.

For an instance $(G,\omega)$ of the (list) channel assignment problem, we have an equivalent instance $(G,\Lambda,t)$ of the generalized list $T$-coloring, where $\Lambda(v)$ is the list of labels available for $v$ (or the set of all labels in a non-list case) and $t(e)=\{0,1,..,\omega(e)-1\}$ for $e \in E(G)$.

\noindent \textbf{$T$-coloring}

The last graph coloring model listed here is so-called $T$-coloring. It has been first introduced by Hale \cite{Hale} as {\em frequency constrained channel assignment problem}. The instance of $T$-coloring is a pair $(G,T)$, where $G$ is a graph and $T$ is a subset of natural numbers. We ask for such a labeling $f$ of vertices $G$ with natural numbers, that $|f(u)-f(v)|\notin T$ for all $uv \in E(G)$. Unlike channel assignment problem, $T$-coloring allows the case when forbidden differences do not form an interval. It is interesting to mention that for $T=\{0,7,14,15\}$ we obtain the model for interferences for the UHF transmitters (see McDiarmid \cite{McDiarmid}). The list version of this problem has been studied for example by  Alon and Zaks \cite{AZ}.

To obtain an instance of the generalized list $T$-coloring, which is equivalent to given instance of the (list) $T$-coloring, we have to set $t(e) = T$ for every $e \in E(G)$.

All the problems mentioned above are NP-complete for general graphs and remain so for many restricted graph classes. Therefore the generalized list $T$-coloring problem is NP-complete as well. When dealing with an NP-hard problem, there is little hope to design an exact algorithm,  running in polynomial time. Therefore we try to design exact exponential algorithm with exponential factor in the complexity bound as small as possible.

Many such algorithms have been presented for the best-studied of the mentioned problems, i.e. graph coloring (see for example Lawler \cite{Lawler}, Eppstein \cite{Eppstein}, Byskov \cite{Byskov}). Currently best exact exact algorithm for graph coloring was presented by Bj\"orklund {\em et al.} \cite{BHK} and runs in time $\mathcal{O}^*(2^n)$\footnote{In the $\mathcal{O}^*$ notation we suppress polynomially bounded terms.}. 

Quite a few algorithms for the $L(2,1)$-labeling problem have also been presented (see Havet {\it et al.} \cite{HKKKL}, Junosza-Szaniawski and Rzążewski \cite{IWOCA, IPL}). Currently best exact algorithm was presented by Junosza-Szaniawski {\em et al.} \cite{TAMC} and has time complexity $\mathcal{O}^*(2.6488^n)$.

An instance of channel assignment is called {\em $\ell$-bounded} if $\omega(e) \leq \ell$ for all $e \in E(G)$. The first exact algorithm for the channel assignment problem with time complexity $\mathcal{O}^*((2\ell+1)^n)$ was proposed by McDiarmid \cite{McD}. Then Kr\'{a}l' \cite{Kral-alg} presented the algorithm running in time $\mathcal{O}^*((\ell+2)^n)$. This bound was beaten by Cygan and Kowalik \cite{CK}, who showed the algorithm with time complexity $\mathcal{O}^*((\ell+1)^n)$.
It still remains a great challenge to design an exact algorithm for the channel assignment problem with time complexity bounded by $\mathcal{O}^*(c^n)$ for $c$ being a constant (or to prove that there is no such algorithm, under standard complexity assumptions).

To our best knowledge, $T$-coloring itself has not raised any attention from the exact algorithms community so far.
In this paper we present a method to solve the generalized $T$-coloring problem. Namely, we adapt the algorithm for the $L(2,1)$-labeling presented by Junosza-Szaniawski {\em et al.} \cite{TAMC}. We focus on the case when $t(e) \neq \{0\}$ for at least one one edge $e$ (in other case we obtain a well-studied list coloring problem, which can be solved in time $\mathcal{O}^*(2^n)$ by adapting the algorithm by Bj\"orklund {\em et al.} \cite{BHK}).  The time and space complexity of our algorithm is bounded by $\mathcal{O}^*((\tau+2)^n)$, where $\tau$ is the maximum forbidden difference over all edges of the input graph.

Although the algorithm by Cygan and Kowalik \cite{CK} is designed for the channel assignment problem, it can be be adapted to solve the generalized $T$-coloring problem. Its time complexity is then the same as the time complexity of our algorithm.
Their approach uses a well-known inclusion-exclusion principle and fast zeta transform.
However, we believe that the method presented in this paper is interesting on its own and can be used to solve many different problems.

In section \ref{sec:complexity} we show that the complexity bound of our algorithm can be improved if the input graph has some special structure. We consider bounded degree graphs, $K_{1,d}$-free graphs (for integer $d$) and graphs having a clique factor, i.e. a spanning subgraph whose every connected component is a clique with at least 2 vertices.

\section{Preliminaries}
\large
An instance of a the generalized list $T$-coloring problem is a triple $(G, \Lambda, t)$, where $G = (V,E)$ is a graph, $\Lambda \colon V \to 2^{\mathbb{N}}$ is a function that assigns  to each vertex a set (list) of permitted labels and $t \colon E \to 2^{\mathbb{N} \cup \{0\}}$ is a function that assigns to each edge a set of forbidden differences over that edge. We assume that $0 \in t(e)$ for any $e \in E$. We aim to find a mapping $\varphi \colon V \to \mathbb{N}$, satisfying the following conditions:
\begin{enumerate}
\item $\varphi(v) \in \Lambda(v)$ for every $v \in V$,
\item $|\varphi(v) - \varphi(w)| \notin t(vw)$ for every edge $vw \in E$.
\end{enumerate}
Such a function is called a {\em proper labeling}.

Let $\Lambda_{max}$ denote the maximum value in the set $\bigcup_{v \in V} \Lambda(v)$. 
We say that an instance of the generalized list $T$-coloring problem is $\tau$-bounded if $\max \{\bigcup_{e \in E} t(e)\} \leq \tau$. 
In this paper we focus on the case when $\tau \geq 1$.

Let $[\tau+1]$ denote the set $\{0,1,..,\tau+1\}$ and $\llbracket \tau + 1 \rrbracket$ denote the set $[\tau+1] \cup \{{\bar 0}\}$, where $\bar 0$ is a special symbol, whose meaning will be made clear later. Note that $| [ \tau + 1 ] | = \tau +2$.

For a vector $\w$ and a set of vectors $A$ let $A_\w$ denote the set $\{ \vv \colon \w\vv\in A \}$ (by $\w\vv$ we denote the concatenation of vectors $\w$ and $\vv$). Vector $\w$ is also called a {\em prefix} of a vector $\w\vv$.

\section{The algorithm}

Let $(G,\Lambda,t)$ be an instance of the general list $T$-coloring problem. We assume that the graph $G$ is connected -- in other case we may label each of its connected components separately.
For the graph $G=(V,E)$ we consider some ordering $v_1,v_2,..,v_n$ of the vertices in $V$ (this ordering will be specified later).

For a partial $k$-labeling $\varphi \colon V \to \{1,\ldots,k\}$ let $\Gamma_{\varphi} \colon V \to \{true, false\}$ be a Boolean function saying, if $\varphi$ can be extended by labeling a particular vertex with $k+1$. Formally, $\Gamma_{\varphi}(v)$ is $true$ if and only if both the following conditions are satisfied:
\begin{enumerate}
\item $k+1 \in \Lambda(v)$,
\item there is no vertex $w \in N(v)$, which is labeled by $\varphi$ with $(k+1) - \varphi(w) \in t(vw)$.
\end{enumerate}

Our strategy is to construct a labeling of $G$ in a way similar to one presented by Junosza-Szaniawski {\em et al.} \cite{TAMC}.

For every $k \in \{1,..,\Lambda_{max}\}$ we introduce a set of vectors $T[k] \subseteq \llbracket \tau + 1 \rrbracket^n$.
The set $T[k]$ contains a vector $\a\in \llbracket \tau + 1 \rrbracket^n$ if and only if
there exists a partial labeling $\varphi\colon V \to \{1,2,..,k\}$
such that :
\begin{enumerate}
\item $\a_i=0$ iff $v_i$ is not labeled by~$\varphi$ and $\Gamma_{\varphi}(v) = true$,
\item $\a_i=\0$ iff $v_i$ is not labeled by~$\varphi$ and $\Gamma_{\varphi}(v) = false$,
\item $\a_i=1$ iff $\varphi(v_i)\leq k - \tau$, and
\item $\a_i=j$ iff $\varphi(v_i)=k+j-\tau -1$ for $j \in \{2,3,..,\tau+1\}$.
\end{enumerate}

In other words, this encoding of partial $k$-channel assignments unifies the sets of vertices labeled with labels not exceeding $k - \tau$, since they do not interfere with the next label to be assigned, which is $k+1$. Moreover, $0$ indicates an unlabeled vertex, which can be labeled with label $k+1$, while $\0$ indicates an unlabeled vertex with label $k+1$ blocked.

Once we have computed $T[\Lambda_{max}]$, we can easily verify if there exists a proper labeling of the input graph -- it corresponds to a vector having no $0$s and $\0$s (as they encode unlabeled vertices). Formally, the instance of the generalized list $T$-coloring problem is a {\sc yes}-instance if and only if $T[\Lambda_{max}] \cap \{1,2,\ldots,\tau,\tau+1\}^n \neq \emptyset$.

The only thing left is to compute the tables $T[k]$ quickly.
Let us define the partial function: $\oplus  \colon \llbracket \tau + 1 \rrbracket \times\{0,1\}\to [\tau +1]$ in the following way:

$$
x \oplus y = \left \{
\begin{array}{l l}
0 & \text{if $x \in \{0,\0\}$ and $y=0$} \\
1 & \text{if $x \in \{1, 2\}$ and $y=0$} \\
x-1 & \text{if $x \in \{3,4,..,\tau+1\}$ and $y=0$} \\
\tau +1 & \text{if $x =0$ and $y=1$} \\
\text{undefined} & \text{otherwise.}
\end{array}
\right .
$$

The table below shows the values of $\oplus$ for different inputs (entry ,,$-$'' means that the value is undefined) .

\begin{displaymath}
\begin{array}{c | c c c c c c c c c}
\oplus & 0 & \0 & 1 & 2 & 3 & 4 & \dots & \tau & \tau +1 \\ \hline
0 & 0 & 0 & 1 & 1 & 2 & 3 & \dots & \tau -1 & \tau \\
1 & \tau +1 & - & - & - & - &- &  \dots & - & -
\end{array}
\end{displaymath}

We generalize $\oplus$ to vectors coordinate-wise, i.e.
$$
x_1x_2.. x_m\oplus y_1y_2.. y_m=
\begin{cases}
(x_1\oplus y_1)..(x_m\oplus y_m)&
\vtop{\hsize=0.3\hsize\noindent\strut
       if $x_i\oplus y_i$ is defined for all
       $i\in\{1,..,m\}$,\strut}\\
\text{undefined}&\text{otherwise.}
\end{cases}
$$

For two sets of vectors $A\subseteq \llbracket \tau + 1 \rrbracket^m$ and 
$B\subseteq\{0,1\}^m$ we define 
$$
A\oplus B= \{\,\a\oplus\b \colon \a\in A,\ \b\in B,\ \a\oplus\b \text{ is
defined}\,\}.
$$

Let $P\subseteq\{0,1\}^n$ be the set of the characteristic vectors of all independent sets of $G$.
Formally, $\p\in P$ if and only if there is an independent set
$X\subseteq V$ such that for all $i \in \{1,2,..,n\}$ holds $\p_i=1$ if and only if $v_i\in X$.

We shall use $T[k] \oplus P$ to find $T[k+1]$.

For a vector $\b \in [ \tau + 1 ]^n$ let $\overline{\b}^{(k)}=\overline{\b}^{(k)}_1\ldots \overline{\b}^{(k)}_n$ denote a vector in $\llbracket \tau + 1 \rrbracket^n$ such that:
$$\overline{\b}^{(k)}_i = \left \{
\begin{array}{l l}
0   & \text{iff $\b_i= 0$ and $k+2 \in \Lambda(v_i)$ and there is no $v_j \in N(v_i)$ such that $\tau -\b_j+ 2 \in t(v_iv_j)$}\\
\0  & \text{iff $\b_i= 0$ and $k+2 \notin \Lambda(v_i)$ or there exists $v_j \in N(v_i)$ such that $\tau -\b_j + 2 \in t(v_iv_j)$}\\
j & \text{iff $\b_i= j$ for $j \in \{1,..,\tau+1\}$.}
\end{array}
\right .$$

Analogously, for a set of vectors $B$ let $\overline B ^{(k)}$ denote the set $\{ \overline{\b}^{(k)} \colon \b \in B\}$.

\begin{lemma}
For all $k \geq 1$ it holds that $\overline{T[k] \oplus P}^{(k)} = T[k+1]$.
\end{lemma}
\begin{proof}
First, we shall prove $\overline{T[k] \oplus P}^{(k)} \subseteq T[k+1]$.

Let $\b'$ be a vector from $\overline{T[k] \oplus P}^{(k)}$. Thus there exist $\b \in T[k]$ and $\p \in P$ such that $\overline{\b \oplus \p}^{(k)} = \b'$.
Let $\varphi$ be a partial $k$-labeling corresponding to $\b$ and $X$ be an independent set encoded by $\p$.

Note that since $\b \oplus \p$ is defined, the labeling $\varphi$ could be extended by labeling every vertex from $X$ with label $k+1$, obtaining a proper partial $(k+1)$-labeling $\varphi'$.
Moreover, $\varphi'$ can be extended to a proper partial $(k+2)$-labeling by labeling a vertex $v_i$ with label $k+2$ if and only if:
\begin{itemize}
\item $v_i$ is not labeled by $\varphi'$ (so $\b_i \in \{0,\0\}$ and $\p_i = 0$),
\item $k+2 \in \Lambda(v)$,
\item there is no neighbor $v_j$ of $v_i$, which is labeled by $\varphi'$ such that $(k+2)- \varphi'(v_j) \in t(v_iv_j)$. Such a conflict cannot occur if $\varphi'(v) \leq k - \tau + 1$, which is equivalent to $\b'_j = 1$ (which happens when $\b_j =1$ or $\b_j = 2$). 
In the remaining cases we have $(k+2)- \varphi'(v_j) \in t(v_iv_j)$ equivalent to $\tau - \b'_j + 2  \in t(v_iv_j)$.
\end{itemize}

To sum up, one can observe that $\b'$ corresponds to a partial $(k+1)$-labeling $\varphi'$ such that:
\begin{enumerate}
\item $\varphi'(v) = \varphi(v)$ iff $\b'_i \notin \{0,\0,\tau+1\}$,
\item $\varphi'(v) = k+1$ iff $v \in X$ (thus $\b'_i = \tau+1$),
\item $v_i$ is unlabeled by $\varphi'$, $k+2 \in \Lambda(v_i)$ and there is no neighbor $w$ of~$v_i$ with $k+2 - \varphi'(w) \in t(v_iw)$ iff $\b'_i=0$,
\item $v_i$ is unlabeled by $\varphi'$, but $k+2 \notin \Lambda(v_i)$ or there is no neighbor $w$ of~$v_i$ with $k+2 - \varphi'(w) \in t(v_iw)$ iff $\b'_i=0$,
\end{enumerate}
Thus $\b' \in T[k+1]$.

On the other hand, let $\b' \in T[k+1]$. Let $\varphi'$ be a partial labeling corresponding to $\b'$ and $\varphi$ be partial labeling defined as follows:
$$\varphi(v) = \left \{
\begin{array}{l l}
\varphi'(v)   & \text{if $\varphi'(v)\neq k+1$,}\\
\text{unlabeled}   & \text{otherwise.}
\end{array}
\right .$$
Since $\varphi$ is a partial $k$-labeling of $G$, there exists $\b \in T[k]$ corresponding to it. Clearly the set $X = \{v \colon \varphi'(v)=k+1\}$ is independent in $G$ and therefore its characteristic vector $\p$ belongs to $P$. Note that $\b' = \overline{\b \oplus \p}^{(k)}$ and therefore $b' \in \overline{T[k] \oplus P}^{(k)}$.
This finishes the proof.
\end{proof}

Moreover, observe that if we set $T[0] := \{0^n\}$, the following holds: $T[1] = \overline{T[0] \oplus P}^{(0)}$.

Note that computing $T[k+1]$ from $T[k] \oplus P$ takes time $\mathcal{O}(|T[k+1]| \cdot n^2)$ (for all pairs of vertices we have to check if they are adjacent and for every vertex we have to check if $k+2$ is on its list of permitted colors).
Having computed all sets $T[k]$, determining the optimal span can be performed in time linear in sizes of these sets.

To compute $T[k] \oplus P$ efficiently we will partition the vertex set into subsets of a bounded size (they shall be defined later). Let $\mathcal{S} = (S_1,S_2,..,S_r)$ be an ordered partition of $V$. Let $s_i := |S_i|$ for all $i \in \{1,..,r\}$ (we require that all $s_i$'s are bounded by some constant $D$). Moreover, the ordering of sets in $\mathcal{S}$ corresponds to the ordering of vertices of the graph (the vertices within each $S_i$ appear in any order):
$$\underbrace{v_1, v_2, \ldots , v_{s_1}}_{S_1}, \underbrace{v_{s_1+1}, v_{s_1+2}, \ldots , v_{s_1+s_2}}_{S_2}, \ldots , \underbrace{v_{n - s_r + 1}, v_{n - s_r + 2}, \ldots , v_{n}}_{S_r}.$$

Then we shall process each of the subsets at once, in the following manner (in each step we consider the first set from $\mathcal{S}$ and delete it from $\mathcal{S}$).
       
$$
T[k]\oplus P = \bigcup_{\substack{\b \in \llbracket \tau + 1 \rrbracket^{s_i}\\ \p \in \{0,1\}^{s_i}\\ \text{s.t. } \b \oplus \p  \text{ is defined}}}  (\b \oplus \p) (T[k]_{\b} \oplus P_{\p}).
$$

We can rewrite this formula in the following way.

$$
T[k]\oplus P = \bigcup_{\substack{\a \in [\tau+1]^{s_i}\\ \p \in \{0,1\}^{s_i}}} \a \left [ \left ( \bigcup
_{\substack{\b \in \llbracket \tau+1 \rrbracket^{s_i}\\ \text{s.t. } \b \oplus \p = \a}} T[k]_{\b}
\right ) \oplus P_{\p} \right ]
$$

The computation can be omitted whenever the prefix $\a$ cannot appear in any vector of $T[k] \oplus P$.
See the pseudo-code of the Algorithm \ref{alg-compute} for another description of the computation of $T[k+1]$.
The input arguments are: a graph $G$, a partition of its vertex set $(S_1,S_2,\ldots,S_r)$, previously computed table $T[k]$ and  the set $P$ of encodings of independent sets in $G$.

\begin{algorithm}[H]
\label{alg-compute}
\caption {\texttt{Compute}$(G,(S_1,S_2,\ldots,S_r),T[k{]},P)$}
$Q \gets \emptyset$\\
\lIf {$r = 0$}
{
\Return $Q$\\
}
\ForEach {$\a \in [\tau+1]^{s_1}$, being a feasible prefix of $T[k] \oplus P$}
{
	$\p \gets$ a binary vector such that $\p_i = 1$ iff $\a_i = \tau+1$ for $1 \leq i \leq s_1$\\
	$A \gets \emptyset$\\
	\ForEach {$b \in \llbracket \tau+1 \rrbracket^{s_1}$ such that $\b \oplus \p = \a$}
	{
		$A \gets A \cup T[k]_{\b}$\\	
	}
	\If {$A \neq \emptyset$}
	{
	$Q' \gets \texttt{Compute}(G,(S_2,\ldots,S_r),A,P_{\p})$\\
	\ForEach{$\vv \in Q'$}
		{$Q \gets Q \cup \{\a\vv\}$}
	}	
}
\Return $Q$
\end{algorithm}

Finally, all sets $T[k]$ are computed and the solution is found by the Algorithm \ref{alg-computeT}. Again, $G$ is the input graph and $\mathcal{S}$ is the partition of its vertex set.

\begin{algorithm}[H]
\label{alg-computeT}
\caption {\algname$(G,\mathcal{S})$}
$P \gets $ a set of characteristic vectors of independent sets of $G$\\
$T[0] \gets \{0^n\}$\\
\For {$k \gets 1$ \KwTo $\Lambda_{max}$}
{
$Q \gets \texttt{Compute}(G,\mathcal{S},T[k-1],P)$\\
$T[k] \gets \overline{Q}^{(k-1)}$\\
}
\lIf {$T[\Lambda_{max}] \cap \{1,2,\ldots,\tau+1\}^n \neq \emptyset$} {\Return {\sc Yes}}\\
\lElse {\Return {\sc No}}
\end{algorithm}

To show where the advantage of processing whole sets $S_i$ at once, consider the following example.
If we process each vertex separately, at each step we have to consider $\tau +2$ prefixes ($0,1,\ldots,\tau+1$).
Suppose now that our input graph has a perfect matching $\mathcal{S}$.
Then, by processing two adjacent vertices at once, we can omit all prefixes in the form $aa$ for $a \in \{2,3,\ldots,\tau+1\}$. This is because each of values in $ \{2,3,\ldots,\tau+1\}$ describes a single label and no two adjacent vertices can get the same label. Therefore instead of considering all $(\tau+2)^2$ prefixes in $[\tau +1] \times [\tau+1]$ (as we would do when processing each vertex separately), we have to deal with $(\tau+2)^2 - \tau$ prefixes. One can observe that the choice of the partition $\mathcal{S}$ is crucial for this approach and will be discussed in Section \ref{sec:complexity}.

To estimate the computational complexity of this approach, we have to calculate the number of pairs $\a$, $\p$, for which there exists at least one $\b$ such that $\b \oplus \p = \a$.
Notice that if we fix $\a$, then $\p_i = 1$ whenever $\a_i = \tau+1$ (otherwise $\p_i = 0$). Therefore the number of such pairs $\a, \p$ is equal to the number of of prefixes $\a$ (corresponding to the vertices in $S_i$), which can appear in a vector in $T[k] \oplus P$ under a partition $\mathcal{S}$. Let $f_i$ denote this number .
In each recursive computation we have to prepare up to $f_i$ pairs of sets of vectors of length $n-s_i$ and compute $\oplus$ on these pairs.
From the result we get the next table $T[k+1]$ in time $\mathcal{O}(|T[k+1]| \cdot n^2)$.
Preparing the recursive calls and combining
their results takes only time linear in the sizes of the tables $T[k]$ and $P$.
The size of $P$ is at most $n \cdot 2^n$ bits and the size of $T[k]$ is at most $n \cdot \prod_{i=1}^r f_i$ bits.
We arrive at the following recursion for the running time (in every step we remove $S_1$ from $\mathcal{S}$ and the index of every remaining set in $\mathcal{S}$ is reduced by one, so that the first one is still called $S_1$):
$$
F(n) =  \mathcal{O} \bigl( n \cdot 2^n + n^3 \cdot \prod _{i=1}^r f_i + f_1 \cdot F(n - s_1)  \bigr).
$$
 
One can verify by induction that this recursion is satisfied by the following formula:
$F(n)= \mathcal{O} \bigl(n^3 \cdot 2^n +  n^3 \cdot \prod _{i=1}^r f_i \bigr )$. By induction hypothesis we have:
 
 \begin{align*}
F(n) = & \mathcal{O} \bigl( n \cdot 2^n + n^3 \cdot \prod _{i=1}^r f_i + f_1 \cdot F(n - s_1) \bigr) \\
	 = & \mathcal{O} \bigl( n \cdot 2^n + n^3 \cdot \prod _{i=1}^r f_i \bigr) + O \bigl(f_1 \cdot (n-s_1)^3 \cdot 2^{n-s_1} + f_1 \cdot (n-s_1)^3 \cdot \prod _{i=2}^r f_i \bigr) \\
	 = & \mathcal{O} \bigl( n \cdot 2^n + n^3 \cdot \prod _{i=1}^r f_i  + f_1 \cdot (n-s_1)^3 \cdot 2^{n-s_1} + (n-s_1)^3 \cdot \prod _{i=1}^r f_i \bigr) \\
	 = & \mathcal{O} \bigl(n^3 \cdot 2^n +  n^3 \cdot \prod _{i=1}^r f_i \bigr ).
\end{align*}

Thus we obtain 
\begin{equation} \label{bound}
F(n)= \mathcal{O}^* \bigl(2^n + \prod _{i=1}^r f_i \bigr ).
\end{equation}

The space complexity of the algorithm is bounded by the total size of sets $T[k]$ and the set $P$. Therefore it is bounded by the same expression as computational complexity, i.e. $\mathcal{O}^* \bigl(2^n + \prod _{i=1}^r f_i \bigr )$.

\subsection{Complexity bounds} \label{sec:complexity}

In this section we shall consider several possible partitions $\mathcal{S}$ of the vertex set and use them to bound the complexity of the algorithm with functions of various invariants of $G$.

\begin{theorem} The algorithm \algname solves the $\tau$-bounded generalized list $T$-coloring problem on a graph $G$ with $n$ vertices in time $\mathcal{O}^*((\tau+2)^n)$.
\end{theorem}

\begin{proof}
Let $v_1,v_2,\ldots,v_n$ be an arbitrary ordering of vertices of $G$ and let $S_i = \{v_i\}$ for $i \in \{1,..,n\}$. Clearly there are at most $\tau+2$ prefixes $\a$ of length $1$, which can appear in any vector from $T[k] \oplus P$. Using formula (\ref{bound}) we obtain the bound for the complexity $F(n)=\mathcal{O}^*(2^n + \prod _{i=1}^n (\tau+2)) = \mathcal{O}^*((\tau+2)^n)$ (recall that we assume that $\tau \geq 1$).
\end{proof}

However, we can improve it by a more careful construction of the partition $\mathcal{S}$.

\subsubsection{Partitions into stars}

Let $\mathcal{S} = \{S_1,S_2,\ldots,S_r\}$ be a partition of the vertex set of $G$, such that for any $i = 1,2,\ldots,r$ we have $s_i \geq 2$ and $G[S_i]$ has a spanning subgraph, which is a star. We call such a partition a {\em star partition} of $G$.
Note that every connected graph (with at least $2$ vertices) has a star partition. Some ways of constructing star partitions will be described further in this section.

\begin{lemma}
Let $G$ be a graph with $n$ vertices having a star partition $\mathcal{S} = \{S_1,S_2,\ldots,S_r\}$ such that $2 \leq s_i \leq D+1$ for some constant $D$ and each $i \in 1,2,\ldots,r-1$ and and $s_r \leq D'$ for some constant $D'$.
The algorithm \algname solves the $\tau$-bounded generalized list $T$-coloring problem on $G$ in time $\mathcal{O}^* \left ((\tau^2 + 3\tau +4)^{n/2} + (2(\tau+2)^{D} + \tau (\tau+1)^{D})^\frac{n}{(D-1)} \right )$.
\end{lemma}

\begin{proof}
Let us consider a subgraph $G[S_i]$ for some $i = 1,2,\ldots,r$. Without a loss of generality let $v_j$ be a central vertex of the spanning star of $G[S_i]$ and $S_i = \{v_j,..,v_{j+s_i-1}\}$.
Note that for any $\a \in T[k] \oplus P$ if $\a_j = h \in \{2,..,\tau+1\}$, then $\a_{j'} \neq h$ for $j'\in \{j+1,..,j+d_i-1\}$, since $h$ represents a single label and each label of any feasible (possibly partial) labeling induces and independent set in $G$.

Hence the number of possible $s_i$-element prefixes of a vector $\a$ in $T[k] \oplus P$ is at most
$$f_i \leq \underbrace{2(\tau+2)^{s_i-1}}_{\a_j \in \{0,1\}} + \underbrace{\tau (\tau+1)^{s_i-1}}_{\a_j \notin \{0,1\}}.$$

This combined with formula (\ref{bound}) gives us the following bound on the complexity:
\begin{align*}
F(n)&=\mathcal{O}^*(2^n + \prod _{i=1}^{r} f_i) \\
&=\mathcal{O}^* \left (2^n + \prod _{i=1}^{r} (2(\tau+2)^{s_i-1} + \tau (\tau+1)^{s_i-1}) \right) \\
&=\mathcal{O}^* \left (2^n + \prod _{i=1}^{r-1} (2(\tau+2)^{s_i-1} + \tau (\tau+1)^{s_i-1})\underbrace{(2(\tau+2)^{s_r-1} + \tau (\tau+1)^{s_r-1})}_{\text{constant}} \right)\\
&=\mathcal{O}^* \left (2^n + \prod _{i=1}^{r-1} (2(\tau+2)^{s_i-1} + \tau (\tau+1)^{s_i-1})\right ). \\
\end{align*}

Note that the value of $\prod _{i=1}^{r-1} (2(\tau+2)^{s_i-1} + \tau (\tau+1)^{s_i-1})$ is maximized if all $s_i$'s are equal. Consider this case and let $d := s_1 = \ldots = s_{r-1}$. Note that $d \leq \frac{n}{r-1}$ and therefore $r-1 \leq n/d$.
Thus we obtain the following:
\begin{align*}
F(n)&=\mathcal{O}^* \left (2^n + \prod _{i=1}^{r-1} (2(\tau+2)^{s_i-1} + \tau (\tau+1)^{s_i-1})\right ) \\
&=\mathcal{O}^* \left (2^n + \prod _{i=1}^{n/d} (2(\tau+2)^{d-1} + \tau (\tau+1)^{d-1})\right )\\
&=\mathcal{O}^* \left (2^n + (2(\tau+2)^{d-1} + \tau (\tau+1)^{d-1})^{n/d}\right ).
\end{align*}

By analyzing the derivative of the function $\alpha_{\tau}(d)=(2(\tau+2)^{d-1} + \tau (\tau+1)^{d-1})^{1/d}$ (defined for $d \geq 2$), one can  observe that for every fixed $\tau$ there exists $d_0$ such that $\alpha_\tau(d)$ is decreasing for $d<d_0$ and increasing for $d > d_0$. Since $d \leq D+1$, we can bound the value of $\alpha_\tau(d)$ by $\max(\alpha_{\tau}(2),\alpha_{\tau}(D+1))$.
Having in mind that $\tau \geq 1$, we obtain the following solution.
\begin{align*}
F(n) &= \mathcal{O}^* \left (\alpha_{\tau}(2)^n+\alpha_{\tau}(D+1)^n \right)\\
&= \mathcal{O}^* \left ((\tau^2 + 3\tau +4)^{n/2} + (2(\tau+2)^{D} + \tau (\tau+1)^{D})^\frac{n}{(D+1)} \right )
\end{align*}
\end{proof}

Observe that for every $\tau$ there exists $d_\tau$ such that $\alpha_\tau(d) > \alpha_\tau(2)$ for any $d \geq d_\tau$ (see Figure \ref{plot-stars}).
\begin{figure}[h]
\begin{center}
\includegraphics[scale=0.5]{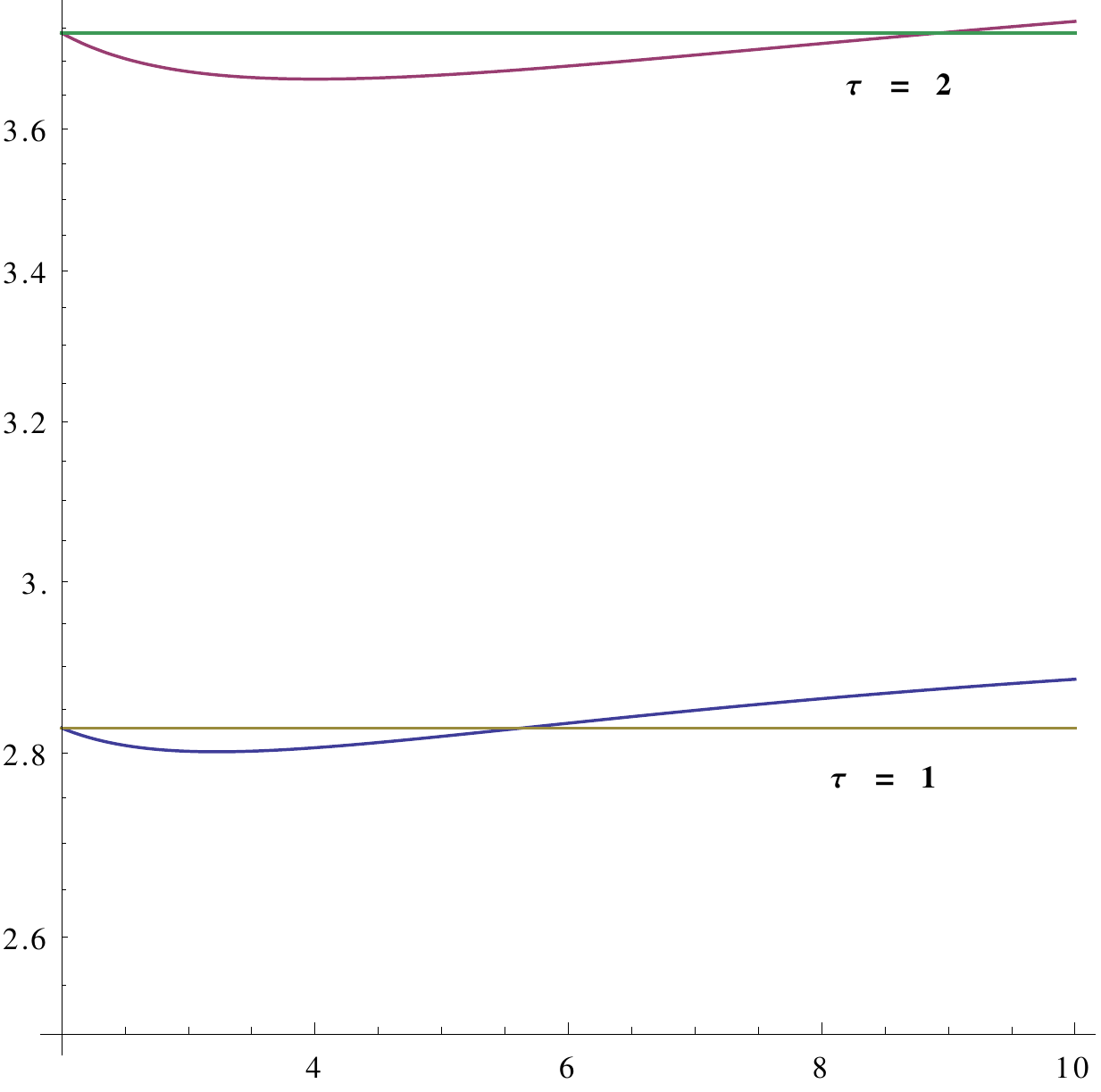}
\end{center}
\caption{The values of $\alpha_\tau(d)$ compared with $\alpha_\tau(2)$ for $\tau=1,2$.}
\label{plot-stars}
\end{figure}

The remainder of this section is devoted to various ways of constructing the initial star factor $\mathcal{S}$.

\noindent \textbf{Remark.} \label{construction-star}
Notice that a star partition can be constructed from a spanning tree of our input graph. Let us consider $T$ being a spanning tree of $G$. Let  $v$ and $u$ be, respectively, the end-vertex and its neighbor on a longest path in $T$. If $T$ is not a star, then all neighbors of $u$ in $T$ except exactly one are leaves in $T$. We include the set $S_i$ consisting of $u$ and all its neighbors which are leaves in $T$ to our partition $\mathcal{S}$. Then we proceed recursively with the tree $T \setminus S_i$. If $T$ is a star, we set $S_r = V(T)$ and finish.

Moreover, notice that if we construct our star partition using spanning tree $T$, in a way described, each set $S_i$ (for $i \in \{1,2,\ldots,r-1\}$) has at most $\Delta(T)$ elements, while $S_r$ has at most $\Delta(T) + 1$ elements.

\vskip 5pt

Observe $\Delta(T) \leq \Delta(G)$ for any spanning tree $T$ of $G$, we obtain the following corollary.

\begin{corollary}
The algorithm \algname solves the $\tau$-bounded generalized list $T$-coloring problem on $G$ with maximum degree bounded by a constant $\Delta$ in time\\ $\mathcal{O}^* \left ((\tau^2 + 3\tau +4)^{n/2} + (2(\tau+2)^{\Delta-1} + \tau (\tau+1)^{\Delta-1})^{n/\Delta} \right )$.
\end{corollary}

The following theorem shows that we may obtain a star partition consisting of smaller stars (and therefore a better bound on the complexity of the algorithm).

\begin{theorem}[Amashi, Kano \cite{AK}, Payan]
Let $D \geq 2$ be an integer such that $\Delta(G) / \delta(G) \leq D$. Then $G$ has a star factor $\mathcal{S} = \{S_1,S_2,\ldots,S_r\}$ such that $s_i \leq D+1$ for any $i =1,2,\ldots,r$.
\end{theorem}

{From} this we can get a significantly better bound for regular graphs (in fact in works for much wider class of graph $G$ with $\Delta(G) \leq 2\delta(G)$).

\begin{corollary}
The algorithm \algname solves the $\tau$-bounded generalized list $T$-coloring problem on a regular graph with $n$ vertices in time $\mathcal{O}^* \left ((\tau^2 + 3\tau +4)^{n/2} \right )$.
\end{corollary}

We can improve this bound for graphs with no big induced stars.
Let $d$ be number, such that $G$ has no induced $K_{1,d}$ star (i.e. it is a $K_{1,d}$-free graph).

The simplest and probably best-studied class with such a property are claw-free graphs (i.e. $K_{1,3}$-free graphs, see for example Brandst\"{a}dt {\em et al.} \cite{Br} for more information). Sumner \cite{Sumner} showed that every claw-free graph with even number of vertices has a perfect matching. Therefore it has a star partition with every set of cardinality 2 (but at most one set with cardinality 3). This observation gives us the following bound.

\begin{corollary}
The algorithm \algname solves the $\tau$-bounded generalized list $T$-coloring problem on a claw-free graph with $n$ vertices in time $\mathcal{O}^* \left ((\tau^2 + 3\tau +4)^{n/2} \right )$.
\end{corollary}

However, we may obtain a similar improvement for $K_{1,d}$-free graph for different $d$.

\begin{lemma}
Every connected $K_{1,d}$-free graph $G$ has a star partition  $\mathcal{S} = \{S_1,S_2,\ldots,S_r\}$ such that $s_i \leq d-1$ for any $i =1,2,\ldots,r-1$ and $s_r \leq d$.
\end{lemma}

\begin{proof}
Again we shall construct $\mathcal{S}$ using a spanning tree. If $G$ has at most $d$ vertices, we set $S_r = V(G)$ and finish. Otherwise, let $T$ be a spanning tree of $G$ and let $v_1$, $u$ and $x$ be, respectively, the first, the second and the third vertex of the longest path in $T$. Notice that $v_1$ is a leaf in $T$ and every neighbor of $u$ in $T$ but at most one vertex (i.e. $x$) is a leaf in $T$ as well. Let $\{v_1, v_2, \ldots, v_k\}$ be the set of neighbors of $u$, which are leaves in $T$. 
We shall consider two cases.

\textbf{Case 1:} If $k \leq d-2$, then we include the set $\{u, v_1, v_2, \ldots, v_k\}$ to our partition and proceed with the tree $T \setminus \{u, v_1, v_2, \ldots, v_k\}$. This set has at most $d-1$ vertices.

\textbf{Case 2:} Suppose now that $k \geq d-1$. Observe that if $k = d-1$, then $x$ is not a leaf, since we assumed that $G$ has at least $d+1$ vertices. Therefore the set $\{u\} \cup \{x, v_1, v_2, \ldots, v_k\}$ does not induce a star in $G$ (as $G$ is $K_{1,d}$-free). Thus $v_iv_j \in E(G)$ for some $1 \leq i < j \leq k$ or $v_ix \in E(G)$ for some $1 \leq i \leq k$. We shall recursively transform our spanning tree $T$ using one of the following transformations.

\begin{itemize}
\item[(T1)] If $v_iv_j \in E(G)$ for some $1 \leq i < j \leq k$ (without loss of generality let $i=1$ and $j=2$), add it to $T$ and remove from $T$ the edge $v_1u$ ($H := (V(G), E(T) \cup \{v_1v_2\} \setminus \{uv_1\}$) (see Figure \ref{trans1}).
\end{itemize}

   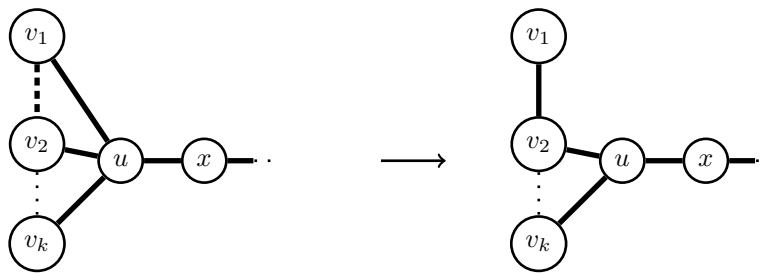
\begin{figure}[h]
  \centering
\begin{tikzpicture}[scale=1.1]
	\tikzstyle{nnode}=[draw, shape=circle, line width=1pt]
    \tikzstyle{edge}=[draw,line width=2pt,-]
    \tikzstyle{arrow}=[draw,line width=1pt,->]
    \tikzstyle{dotedge}=[draw,line width=2pt, densely dashed]
    \tikzstyle{ldotedge}=[draw,line width=1pt, loosely dotted]
   
		\node[nnode] (v1) at (0,1.5) {$v_1$};
		\node[nnode] (v2) at (0,0.2) {$v_2$};
		\node[nnode] (vk) at (0,-1) {$v_k$};
		\node[nnode] (u) at (1,0) {$u$};	
        \node[nnode] (x) at (2,0) {$x$};	
        \node (rest) at (3,0) { };	
        \node (rest1) at (2.7,0) { };	
        \node (from) at (4,0) { };	
        \draw [edge] (v1) -- (u);
        \draw [edge] (v2) -- (u);
        \draw [edge] (vk) -- (u);
        \draw [edge] (u) -- (x);
    	\draw [dotedge] (v1) to (v2);
        \draw [edge] (x) -- (rest1);
        \draw [ldotedge] (x) -- (rest);		
        \draw [ldotedge] (v2) -- (vk);		
        \node (to) at (5,0) { };	
		\node[nnode] (v1a) at (6,1.5) {$v_1$};
		\node[nnode] (v2a) at (6,0.2) {$v_2$};
		\node[nnode] (vka) at (6,-1) {$v_k$};
		\node[nnode] (ua) at (7,0) {$u$};	
        \node[nnode] (xa) at (8,0) {$x$};	
        \node (resta) at (9,0) { };	
        \node (rest1a) at (8.7,0) { };	
        \draw [edge] (v2a) -- (ua);
        \draw [edge] (vka) -- (ua);
        \draw [edge] (ua) -- (xa);
    	\draw [edge] (v1a) to (v2a);
        \draw [edge] (xa) -- (rest1a);
        \draw [ldotedge] (xa) -- (resta);		
        \draw [ldotedge] (v2a) -- (vka);
       \draw [arrow] (from) -- (to);		

\end{tikzpicture}
\caption{\label{trans1} Transformation (T1)}
\end{figure}

\begin{itemize}
\item[(T2)] If $v_ix \in E(G)$ for some $1 \leq k$ (without loss of generality let $i=1$), add it to $T$ and remove from $T$ the edge $v_1u$ ($H := (V(G), E(T) \cup \{v_1x\} \setminus \{uv_1\}$) (see Figure \ref{trans2}).
\end{itemize}

   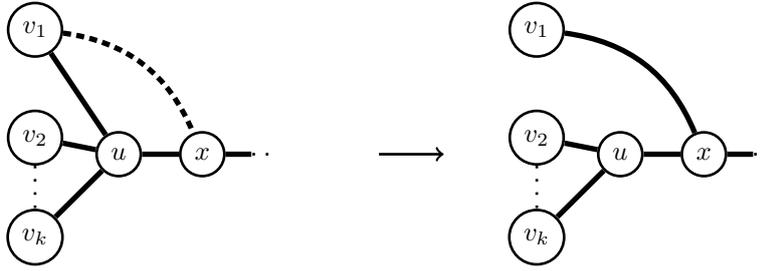
\begin{figure}[h]
  \centering
\begin{tikzpicture}[scale=1.1]
	\tikzstyle{nnode}=[draw, shape=circle, line width=1pt]
    \tikzstyle{edge}=[draw,line width=2pt,-]
    \tikzstyle{arrow}=[draw,line width=1pt,->]
    \tikzstyle{dotedge}=[draw,line width=2pt, densely dashed]
    \tikzstyle{ldotedge}=[draw,line width=1pt, loosely dotted]
   
		\node[nnode] (v1) at (0,1.5) {$v_1$};
		\node[nnode] (v2) at (0,0.2) {$v_2$};
		\node[nnode] (vk) at (0,-1) {$v_k$};
		\node[nnode] (u) at (1,0) {$u$};	
        \node[nnode] (x) at (2,0) {$x$};	
        \node (rest) at (3,0) { };	
        \node (rest1) at (2.7,0) { };	
        \node (from) at (4,0) { };	
        \draw [edge] (v1) -- (u);
        \draw [edge] (v2) -- (u);
        \draw [edge] (vk) -- (u);
        \draw [edge] (u) -- (x);
    	\draw [dotedge] (v1) to [bend left] (x);
        \draw [edge] (x) -- (rest1);
        \draw [ldotedge] (x) -- (rest);		
        \draw [ldotedge] (v2) -- (vk);		
        \node (to) at (5,0) { };	
		\node[nnode] (v1a) at (6,1.5) {$v_1$};
		\node[nnode] (v2a) at (6,0.2) {$v_2$};
		\node[nnode] (vka) at (6,-1) {$v_k$};
		\node[nnode] (ua) at (7,0) {$u$};	
        \node[nnode] (xa) at (8,0) {$x$};	
        \node (resta) at (9,0) { };	
        \node (rest1a) at (8.7,0) { };	
        \draw [edge] (v2a) -- (ua);
        \draw [edge] (vka) -- (ua);
        \draw [edge] (ua) -- (xa);
    	\draw [edge] (v1a) to [bend left] (xa);
        \draw [edge] (xa) -- (rest1a);
        \draw [ldotedge] (xa) -- (resta);		
        \draw [ldotedge] (v2a) -- (vka);
       \draw [arrow] (from) -- (to);		

\end{tikzpicture}
\caption{\label{trans2} Transformation (T2)}
\end{figure}

Notice that if none of the above transformations can be applied to $T$, then $k \leq d-2$ and this is exactly Case 1.
This finishes the proof.
\end{proof}

Having such a partition, we obtain the following complexity bound.

\begin{corollary}
The algorithm \algname solves the $\tau$-bounded generalized list $T$-coloring problem on a $K_{1,d}$-free graph with $n$ vertices in time\\ $\mathcal{O}^*\left ((\tau^2 + 3\tau +4)^{n/2} + (2(\tau+2)^{d-2} + \tau (\tau+1)^{d-2})^{n/(d-1)} \right )$.
\end{corollary}

Another class of graphs we want to mention are unit disk graphs, i.e. intersection graphs of unit disks on a plane (see for example Clark {\em et al.} \cite{UDG}). They are particularly interesting due to their applications in modeling of ad-hoc networks. It is widely known that unit disk graphs are $K_{1,7}$-free. 

\begin{corollary}
The algorithm \algname solves the $\tau$-bounded generalized list $T$-coloring problem on a unit disk graph with $n$ vertices in time\\ $\mathcal{O}^* \left ((\tau^2 + 3\tau +4)^{n/2} + (2(\tau+2)^{5} + \tau (\tau+1)^{5})^{n/6} \right )$.
\end{corollary}

Recall that star partitions consisting of small stars yield lower complexity bound of the algorithm. Therefore, if we want to construct them using a spanning tree, the maximum degree of such a tree should be lowest possible. However, deciding if the input graph has a spanning tree with maximum degree at most $k$ is NP-complete for every $k \geq2$ (see Garey, Johnson \cite{GJ}).
\subsubsection{Partitions into cliques}

Another subgraphs in the input graph may also prove useful in constructing the partition $\mathcal{S}$.
Hell and Kirkpatrick \cite{HK} studied the problem of partitioning the vertex set of a graph into cliques.
Let \textit{clique packing} be a subgraph whose every connected component is a clique with at least 2 vertices.
Let $\rho(G)$ denote the order of the largest (in terms of the number of vertices)  clique packing in $G$. Note that a matching is a special case of a clique packing of $G$. Therefore, if $m$ is the size of maximum matching in $G$, we have: $2m \leq \rho(G)$.

Let $G$ be a graph, such that $|V(G)| = \rho(G)$ and let $H$ be its largest clique packing. A {\em  clique partition} is a partition of $V(G)$ into vertex sets of connected components of $H$. In other words, every set of the clique partition induces a clique in $G$.
 
Hell and Kirkpatrick \cite{HK} presented an elegant structural theorem allowing to compute the value of $\rho(G)$.
Moreover, they described graphs $G$ having a clique partition.

\begin{theorem}
The algorithm \algname solves the $\tau$-bounded generalized list $T$-coloring problem on a graph $G$ with $n$ vertices in time $\mathcal{O}^* \left ((\tau^2 + 3\tau +4)^{\rho(G)/2} (\tau+2)^{n-\rho(G)}  \right)$.
\end{theorem}

\begin{proof}
Let $H$ be the largest clique packing of $G$. Consider a spanning subgraph $H'$ of $H$ in which every connected component is isomorphic to $K_2$ or $K_3$. It can be clearly done -- the vertex set of connected component of $H$ with even number of vertices is partitioned into single edges (a perfect matching), while each component with odd number of vertices is partitioned into single edges and exactly one triangle.

Let $p$ be a number of components isomorphic to $K_2$ and $q$ be the number of components isomorphic to $K_3$. Clearly $2p+3q = \rho(G)$.
Let $\mathcal{S}=\{S_1, .., S_r\}$ be a partition of $V(G)$, such that the sets $S_i$ for $i \leq p$ correspond to $K_2$-components in $H'$ and the sets $S_i$ for $p \leq p+q$ correspond to $K_3$-components of $H'$. The remaining sets $S_i$ for $i > c$ contain the remaining vertices of $G$, one vertex per set.

Let us consider a subgraph $G[S_i]$ for some $i \leq p + q$.  Since it is a clique, each vertex has to be labeled with a different label. Therefore in any prefix of a vector from $T[k] \oplus P$, corresponding to the vertices from $S_i$, each element from $\{2,..,\tau+1\}$ may appear at most once. Recall that for $S_i$ inducing an edge (i.e. for $i \leq p$) we have at most $\tau^2 + 3\tau +4$ possible prefixes.

Now let us consider 3-element prefix $\a$ corresponding to a triangle. There are:
\begin{itemize}
\item $2^3=8$ prefixes with no element from  $\{2,..,\tau+1\}$,
\item $3\tau \cdot 2^2 = 12 \cdot \tau$ prefixes having exactly one element from $\{2,..,\tau+1\}$,
\item $3\tau (\tau-1) \cdot 2 = 6\cdot \tau(\tau-1)$ prefixes having exactly two elements from $\{2,..,\tau+1\}$,
\item $\tau (\tau-1) (\tau-2)$ prefixes having exactly three elements from $\{2,..,\tau+1\}$.
\end{itemize}
For the remaining sets $S_i$ (containing single vertices, i.e. for $i>p+q$), there are $\tau+2$ possible prefixes of length 1.
Hence the number of possible $s_i$-element prefixes of a vector $\a$ in $T[k] \oplus P$ is at most
$$
f_i \leq 
\begin{cases}
\tau^2 + 3\tau +4 & \text{ for } i \leq p\\
\tau^3 + 3\tau^2 + 8\tau +8 & \text{ for } p < i \leq p + q\\
\tau+2 & \text{ for } i  > p+q.\\
\end{cases}$$

Again, using formula (\ref{bound}), we obtain the following.

\begin{align*}
F(n)&=\mathcal{O}^*(2^n + \prod _{i=1}^{r} f_i) \\
&=\mathcal{O}^* \left (2^n + (\tau^2 + 3\tau +4)^p (\tau^3 + 9\tau^2 + 2\tau +8)^q (\tau+2)^{n-\rho(G)}  \right) \\
&=\mathcal{O}^* \left (2^n + (\tau^2 + 3\tau +4)^\frac{\rho(G)-3q}{2} (\tau^3 + 3\tau^2 + 8\tau +8)^q (\tau+2)^{n-\rho(G)}  \right) \\
\end{align*}

This expression is maximized for $q=0$. So finally we obtain the bound:

\begin{align*}
F(n) &=\mathcal{O}^* \left (2^n + (\tau^2 + 3\tau +4)^{\rho(G)/2} (\tau+2)^{n-\rho(G)}  \right)\\
	 &=\mathcal{O}^* \left ((\tau^2 + 3\tau +4)^{\rho(G)/2} (\tau+2)^{n-\rho(G)}  \right)\\
\end{align*}
\end{proof}

The Table \ref{comparison} compares the complexity bounds (more precisely, the bases of the exponential factor) of the algorithm \algname applied to various graph classes for some values of $\tau$.

\begin{table}[h]
\begin{center}
\begin{tabular}{|l||l|l|l|l|l|l|}
\hline
	 &	general	& subcubic	& claw-free		&	regular	&	graphs having &	unit disk \\
$\tau$		 &	graphs	& graphs			&	graphs		&	graphs	&	clique partition & graphs	 \\ \hline \hline
1		 &	3.0000	& 2.8021		& \multicolumn{3}{c|}{2.8285}							& 2.8340	 \\ \hline
2		 &	4.0000	& 3.6841		& \multicolumn{4}{c|}{3.7417}  \\ \hline
3		 &	5.0000	& 4.6105		& \multicolumn{4}{c|}{4.6905} \\ \hline
4		 &	6.0000	& 5.5613		& \multicolumn{4}{c|}{5.6569}		 \\ \hline
5		 &	7.0000	& 6.5266		& \multicolumn{4}{c|}{6.6333}		 \\ \hline
\end{tabular}
\end{center}
\caption{Comparison of bases of the exponential factor in the complexity bound.}
\label{comparison}
\end{table}

\noindent \textbf{Acknowledgement.} The authors are grateful to Professor Zbigniew Lonc for useful advice on graph factors.


\begin{thebibliography}{10}

\bibitem{AZ}
N.~Alon and A.~Zaks.
\newblock T-choosability in graphs.
\newblock {\em Discrete Applied Mathematics}, 82:1 -- 13, 1998.

\bibitem{AK}
A.~Amahashi and M.~Kano.
\newblock On factors with given components.
\newblock {\em Discrete Mathematics}, 42:1 -- 6, 1982.

\bibitem{BHK}
A.~Bj\"{o}rklund, T.~Husfeldt, and M.~Koivisto.
\newblock Set partitioning via inclusion-exclusion.
\newblock {\em SIAM Journal on Computing}, 39:546--563, 2009.

\bibitem{Br}
A.~Brandst\"{a}dt, V.~B. Le, and J.~P. Spinrad.
\newblock {\em Graph classes: a survey}.
\newblock Society for Industrial and Applied Mathematics, Philadelphia, PA,
  USA, 1999.

\bibitem{Byskov}
J.~M. Byskov.
\newblock Enumerating maximal independent sets with applications to graph
  colouring.
\newblock {\em Operations Research Letters}, 32:547–556, 2004.

\bibitem{Calamoneri}
T.~Calamoneri.
\newblock The {L(h,k)}-labelling problem: An updated survey and annotated
  bibliography.
\newblock {\em The Computer Journal}, 54:1344--1371, 2011.

\bibitem{UDG}
B.~N. Clark, C.~J. Colbourn, and D.~S. Johnson.
\newblock Unit disk graphs.
\newblock {\em Discrete Mathematics}, 86:165 -- 177, 1990.

\bibitem{CK}
M.~Cygan and L.~Kowalik.
\newblock Channel assignment via fast zeta transform.
\newblock {\em Information Processing Letters}, 111:727--730, 2011.

\bibitem{Eppstein}
D.~Eppstein.
\newblock Small maximal independent sets and faster exact graph coloring.
\newblock {\em Journal on Graph Algorithms and Applications}, 7:131--140, 2003.

\bibitem{ERT}
P.~Erd\"os, A.~Rubin, and H.~Taylor.
\newblock Choosability in graphs.
\newblock {\em Proc. West Coast Conference on Combinatorics, Graph Theory and
  Computing, Arcata, Congressus Numerantium}, 26:125 -- 157, 1979.

\bibitem{FKS}
J.~Fiala, D.~Kr\'{a}l', and R.~Skrekovski.
\newblock A brooks-type theorem for the generalized list t-coloring.
\newblock {\em SIAM Journal on Discrete Mathematics}, 19:588--609, 2005.

\bibitem{FS}
J.~Fiala and R.~\v{S}krekovski.
\newblock List distance labelings of graphs.
\newblock {\em KAM Series}, 530, 2001.

\bibitem{GJ}
M.~R. Garey and D.~S. Johnson.
\newblock {\em Computers and Intractability; A Guide to the Theory of
  NP-Completeness}.
\newblock W. H. Freeman \& Co., New York, NY, USA, 1990.

\bibitem{GY}
J.~Griggs and R.~K. Yeh.
\newblock Labeling graphs with a condition at distance two.
\newblock {\em SIAM Journal on Discrete Mathematics}, 5:586 -- 595, 1992.

\bibitem{Hale}
W.~Hale.
\newblock Frequency assignment: theory and applications.
\newblock {\em Proc. 1EEE}, 68:1497--1514, 1980.

\bibitem{HKKKL}
F.~Havet, M.~Klazar, J.~Kratochv{\'{i}}l, D.~Kratsch, and M.~Liedloff.
\newblock Exact algorithms for l(2,1)-labelling.
\newblock {\em Algorithmica}, 59:169--194, 2011.

\bibitem{HK}
P.~Hell and D.~G. Kirkpatrick.
\newblock Packings by cliques and by finite families of graphs.
\newblock {\em Discrete Mathematics}, 49:45--59, 1984.

\bibitem{JT}
T.~R. Jensen and B.~Toft.
\newblock {\em Graph Coloring Problems}.
\newblock John Wiley and Sons, 1995.

\bibitem{TAMC}
K.~Junosza-Szaniawski, J.~Kratochvíl, M.~Liedloff, P.~Rossmanith, and
  P.~Rz\k{a}\.{z}ewski.
\newblock Fast exact algorithm for {L(2,1)}-labeling of graphs.
\newblock {\em Theoretical Computer Science}, 505:42 -- 54, 2013.

\bibitem{IWOCA}
K.~"Junosza-Szaniawski and P.~Rz\k{a}\.{z}ewski.
\newblock On improved exact algorithms for {L(2,1)}-labeling of graphs.
\newblock {\em Proc. of IWOCA 2010, Lecture Notes in Computer Science},
  "6460":"34--37", "2011".

\bibitem{IPL}
K.~Junosza-Szaniawski and P.~Rz\k{a}\.{z}ewski.
\newblock On the complexity of exact algorithm for {L(2, 1)}-labeling of
  graphs.
\newblock {\em Information Processing Letters}, 111:697 -- 701, 2011.

\bibitem{Kral-alg}
D.~Kr\'{a}l'.
\newblock An exact algorithm for the channel assignment problem.
\newblock {\em Discrete Applied Mathematics}, 145:326--331, 2005.

\bibitem{Kral-surv}
D.~Kr\'{a}l'.
\newblock The channel assignment problem with variable weights.
\newblock {\em SIAM Journal on Discrete Mathematics}, 20:690--704, 2006.

\bibitem{Lawler}
E.~Lawler.
\newblock A note on the complexity of the chromatic number problem.
\newblock {\em Information Processing Letters}, 5:66--67, 1976.

\bibitem{McDiarmid}
C.~McDiarmid.
\newblock Discrete mathematics and radio channel assignment.
\newblock {\em CMS Books Math. Ouvrages Math. SMC}, 11:27–63, 2003.

\bibitem{McD}
C.~McDiarmid.
\newblock On the span in channel assignment problems: bounds, computing and
  counting.
\newblock {\em Discrete Mathematics}, 266:387 -- 397, 2003.

\bibitem{Sumner}
D.~Sumner.
\newblock Graphs with 1-factors.
\newblock {\em Proceedings of the AMS}, 42:8 -- 12, 1974.

\bibitem{Vizing}
V.~Vizing.
\newblock Vertex colorings with given colors.
\newblock {\em Metody Diskret. Analiz. (in Russian)}, 29:3 -- 10, 1976.

\bibitem{Yeh}
R.~K. Yeh.
\newblock A survey on labeling graphs with a condition at distance two.
\newblock {\em Discrete Mathematics}, 306:1217 -- 1231, 2006.

\end{thebibliography}
\end{document}